\documentclass[orivec,conference]{llncs}
\usepackage[cmex10]{amsmath}
\usepackage{stmaryrd}
\usepackage{amsfonts}
\usepackage{amssymb}
\usepackage{bbold}
\usepackage{wrapfig}
\usepackage{graphicx}
\usepackage{enumerate}
\usepackage{multirow}
\usepackage{xspace}
\usepackage{cancel}
\usepackage{enumitem}
\usepackage[hidelinks]{hyperref}
\usepackage{esvect}
\usepackage{caption}
\usepackage{xcolor,colortbl}

\newcommand{\Eg}{\mbox{$\mi{EG}$-$\mi{PQE}$}\xspace}
\newcommand{\egp}{\mbox{$\mi{EG}$-$\mi{PQE}^+$}\xspace}

\newcommand{\Vp}{\mbox{$\mi{VerPQE}$}\xspace}
\newcommand{\Cr}{\mbox{$\mi{ChkRed}$}\xspace}

\newcommand{\bm}[1]{{\mbox{\boldmath $#1$}}}
\newcommand{\Bm}[1]{{\boldmath $#1$}}

\newcommand{\pl}{\mbox{$\mi{Plg}$}\xspace}

\newcommand{\imp}{\Rightarrow} 

\newcommand{\pnt}[1]{\mbox{$\vv{#1}$}\xspace}

\newcommand{\cof}[2]{\mbox{$#1_{\vec{#2}}$}}

\newcommand{\V}[1]{\mbox{$\mathit{Vars}(#1)$}}
\newcommand{\Va}[1]{\mbox{$\mi{Vars}(\vec{#1})$}}

\newcommand{\s}[1]{\mbox{$\{#1\}$}}

\newcommand{\nGz}[2]{$G_{non-\{z\}}$}

\newcommand{\prr}[1]{\mi{Prev}(\boldsymbol{q})}

\newcommand{\mi}[1]{\mathit{#1}}
\newcommand{\ti}[1]{\textit{#1}}
\newcommand{\tb}[1]{\textbf{#1}}

\newcommand{\ttt}{\>\>\>}

\newcommand{\Tt}{\>\>}

\newcommand{\prob}[2]{\mbox{$\exists{#1} [#2]$}}

\newcommand{\Comment}[1]{}

\begin{document}

\title{Verification Of Partial Quantifier Elimination}


\author{Eugene Goldberg}
\institute{\email{eu.goldberg@gmail.com}}

\maketitle

\begin{abstract}
Quantifier elimination (QE) is an important problem that has numerous
applications. Unfortunately, QE is computationally very hard. Earlier
we introduced a generalization of QE called \ti{partial} QE (or PQE
for short). PQE allows to unquantify a \ti{part} of the formula.  The
appeal of PQE is twofold. First, many important problems can be solved
in terms of PQE. Second, PQE can be drastically faster than QE if only
a small part of the formula gets unquantified.  To make PQE practical,
one needs an algorithm for verifying the solution produced by a PQE
solver. In this paper, we describe a very simple SAT-based verifier
called \Vp and provide some experimental results.
\end{abstract}

\section{Introduction}
Earlier, we introduced a generalization of Quantifier Elimination (QE)
called \ti{partial} QE (or PQE for short)~\cite{hvc-14}.  PQE allows
to unquantify a \ti{part} of the formula. So, QE is just a special
case of PQE where the entire formula gets unquantified.  The appeal of
PQE is twofold. First, it can be much more efficient than QE if only a
small part of the formula gets unquantified. Second, many known
verification problems like SAT, equivalence checking, model checking
and new problems like property generation can be solved in terms of
PQE~\cite{hvc-14,south_korea,fmcad16,mc_no_inv2,eg_pqe_tech}. So, PQE
can be used to design new efficient algorithms. To make PQE practical,
one needs to verify the correctness of the solution provided by a PQE
solver. Such verification is the focus of this paper.

We consider PQE on propositional formulas in conjunctive normal form
(CNF)\footnote{Every formula is a propositional CNF formula unless otherwise
stated. Given a CNF formula $F$ represented as the conjunction of
clauses $C_1 \wedge \dots \wedge C_k$, we will also consider $F$ as
the \ti{set} of clauses \s{C_1,\dots,C_k}.
} with existential quantifiers. PQE is
defined as follows. Let $F(X,Y)$ be a propositional CNF formula where
$X,Y$ are sets of variables. Let $G$ be a subset of clauses of $F$.
Given a formula \prob{X}{F}, find a quantifier-free formula $H(Y)$
such that $\prob{X}{F}\equiv H\wedge\prob{X}{F \setminus G}$.  In
contrast to QE, only the clauses of $G$ are taken out of the scope of
quantifiers here (hence the name partial QE).  We will refer to $H$ as
a \tb{solution} to PQE. As we mentioned above, PQE \ti{generalizes}
QE. The latter is just a special case of PQE where $G = F$ and the
entire formula is unquantified.

To verify the solution $H$ above one needs to check if
$\prob{X}{F}\equiv H\wedge\prob{X}{F \setminus G}$ indeed holds. If
derivation of $H$ is done in some proof system, one can check the
correctness of $H$ by verifying the proof (like it is done for
SAT-solvers). Since, PQE is currently in its infancy and no well
established proof system exists we use a more straightforward
approach.  Namely, we present a very simple SAT-based verification
algorithm called \Vp that does not require any knowledge of how the
solution $H$ is produced. A flaw of \Vp is that, in general, it does
not scale well. Nevertheless, \Vp can be quite useful in two
scenarios. First, \Vp is efficient enough to handle PQE problems formed
from random formulas of up to 70-80 variables. Such examples can can
be employed when debugging a PQE solver. Second, \Vp can efficiently
verify even large PQE problems for a particular class of formulas
described in Subsection~\ref{ssec:perf}.

The paper is structured as follows. Basic definitions are given in
Section~\ref{sec:basic}. Section~\ref{sec:ver_pqe} formally describes
how a solution to PQE can be verified. The verification algorithm
called \Vp is presented in
Section~\ref{sec:ver_alg}. Section~\ref{sec:expers} gives experimental
results. Some background is provided in Section~\ref{sec:bkgr} and
conclusions are made in Section~\ref{sec:concl}.

\vspace{-2pt}
\section{Basic Definitions}
\vspace{-1pt}
\label{sec:basic}

In this section, when we say ``formula'' without mentioning
quantifiers, we mean ``a quantifier-free formula''.

\begin{definition}
\label{def:cnf}
We assume that formulas have only Boolean variables.  A \tb{literal}
of a variable $v$ is either $v$ or its negation.  A \tb{clause} is a
disjunction of literals. A formula $F$ is in conjunctive normal form
(\tb{CNF}) if $F = C_1 \wedge \dots \wedge C_k$ where $C_1,\dots,C_k$
are clauses. We will also view $F$ as a \tb{set of
clauses} \s{C_1,\dots,C_k}. We assume that \tb{every formula is in
CNF} unless otherwise stated.
\end{definition}

%
%
\begin{definition}
  \label{def:vars} Let $F$ be a formula. Then \bm{\V{F}} denotes the
set of variables of $F$ and \bm{\V{\prob{X}{F}}} denotes
$\V{F}\!\setminus\!X$.
\end{definition}

%
%
\begin{definition}
Let $V$ be a set of variables. An \tb{assignment} \pnt{q} to $V$ is a
mapping $V'~\rightarrow \s{0,1}$ where $V' \subseteq V$.  We will
denote the set of variables assigned in \pnt{q}~~as \bm{\Va{q}}. We will
refer to \pnt{q} as a \tb{full assignment} to $V$ if $\Va{q}=V$. We
will denote as \bm{\pnt{q} \subseteq \pnt{r}} the fact that a) $\Va{q}
\subseteq \Va{r}$ and b) every variable of \Va{q} has the same value
in \pnt{q} and \pnt{r}.
\end{definition}

%
%
\begin{definition}
A literal and a clause are said to be \tb{satisfied}
(respectively \tb{falsified}) by an assignment \pnt{q} if they
evaluate to 1 (respectively 0) under \pnt{q}.
\end{definition}

%
%
\begin{definition}
\label{def:cofactor}
Let $C$ be a clause. Let $H$ be a formula that may have quantifiers,
and \pnt{q} be an assignment to
\V{H}.  If $C$ is satisfied by \pnt{q}, then \bm{\cof{C}{q} \equiv
  1}. Otherwise, \bm{\cof{C}{q}} is the clause obtained from $C$ by
removing all literals falsified by \pnt{q}. Denote by \bm{\cof{H}{q}}
the formula obtained from $H$ by removing the clauses satisfied by
\pnt{q} and replacing every clause $C$ unsatisfied by \pnt{q} with
\cof{C}{q}.
\end{definition}

%
%
\begin{definition}
  \label{def:Xcls}
Given a formula \prob{X}{F(X,Y)}, a clause $C$ of $F$ is called
\tb{quantified} if \V{C} $\cap~X~\neq~\emptyset$. 
\end{definition}

%
%
\begin{definition}
\label{def:formula-equiv}
Let $G, H$ be formulas that may have existential quantifiers. We say
that $G, H$ are \tb{equivalent}, written \bm{G \equiv H}, if
$\cof{G}{q} =
\cof{H}{q}$ for all full assignments \pnt{q} to $\V{G} \cup \V{H}$.
\end{definition}

%
%
\begin{definition}
\label{def:red_cls}
Let $F(X,Y)$ be a formula and $G \subseteq F$ and $G \neq
\emptyset$. The clauses of $G$ are said to be \textbf{redundant in} \bm{\prob{X}{F}} if
$\prob{X}{F} \equiv \prob{X}{F \setminus G}$.
\end{definition}

%
%
\begin{definition}
 \label{def:pqe_prob} Given a formula \prob{X}{F(X,Y))} and $G$ where
 $G \subseteq F$, the \tb{Partial Quantifier Elimination} (\tb{PQE})
 problem is to find $H(Y)$ such that\linebreak \Bm{\prob{X}{F}\equiv
 H\wedge\prob{X}{F \setminus G}}.  (So, PQE takes $G$ out of the scope
 of quantifiers.)  The formula $H$ is called a \tb{solution} to
 PQE. The case of PQE where $G = F$ is called \tb{Quantifier
 Elimination} (\tb{QE}).
\end{definition}
\begin{remark}
For the sake of simplicity, we will assume that every clause of formula
$G$ in Definition~\ref{def:pqe_prob} is quantified.
\end{remark}
%
%
\begin{example}
\label{exmp:pqe_exmp}
Consider the formula $F = C_1 \wedge C_2 \wedge C_3 \wedge C_4$ where
$C_1=\overline{x}_3 \vee x_4$, $C_2\!=\!y_1\!\vee\!x_3$,
$C_3=y_1 \vee \overline{x}_4$, $C_4\!=\!y_2\!\vee\!x_4$. Let $Y$
denote \s{y_1,y_2} and $X$ denote \s{x_3,x_4}. Consider the PQE
problem of taking $C_1$ out of \prob{X}{F} i.e. finding $H(Y)$ such
that $\prob{X}{F} \equiv H \wedge \prob{X}{F \setminus \s{C_1}}$. One
can show that $\prob{X}{F} \equiv
y_1 \wedge \prob{X}{F \setminus \s{C_1}}$.  That is, $H\!  =\!y_1$ is
a solution to the PQE problem above.
\end{example}

%
%
\begin{proposition}
\label{prop:sol_impl}
Let $H$ be a solution to the PQE problem of
Definition~\ref{def:pqe_prob}.  That is $\prob{X}{F}\equiv
H\wedge\prob{X}{F \setminus G}$. Then $F \imp H$ (i.e. $F$ implies $H$).
\end{proposition}
The proofs of propositions are given in Appendix~\ref{app:proofs}.

\section{Verification Of A Solution to PQE}
\label{sec:ver_pqe}

Let $H(Y)$ be a solution found by a PQE solver when taking $G$ out of
\prob{X}{F(X,Y)}. That is $\prob{X}{F} \equiv H \wedge \prob{X}{F
  \setminus G}$ is supposed to hold. One can check if this is true
(i.e. whether $H$ is correct) using the proposition below.

\begin{proposition}
\label{prop:main}
Formula $H(Y)$ is a solution to the PQE problem of taking $G$ out
of the scope of quantifiers in \prob{X}{F(X,Y)} if and only if
  \begin{enumerate}
  \item[a)]  $H$ is implied by $F$;
  \item[b)]  $G$ is redundant in $H \wedge \prob{X}{F}$ i.e.
    $H \wedge \prob{X}{F} \equiv H \wedge \prob{X}{F \setminus G}$
 \end{enumerate}
\end{proposition}

Checking the first condition of Proposition~\ref{prop:main} can be
done by a SAT-solver. Namely, one just needs to check for every clause
$C$ of $H$ if $F \wedge \overline{C}$ is unsatisfiable. (If so, then
$F \imp C$.) Below, we describe how one can check the second condition
of Proposition~\ref{prop:main} in terms of boundary points.

%
%
\begin{definition}
\label{def:bnd_pnt}
  Let $F$ be a formula and $G$ be a non-empty subset of clauses of
  $F$. A full assignment \pnt{p} to \V{F} is called a
  \bm{G}\tb{-boundary point} of $F$ if it falsifies $G$ and satisfies
  $F \setminus G$.
\end{definition}

The name ``boundary point'' is due to the fact that if the subset $G$
is small,\pnt{p} can sometimes be close to the boundary between
assignments satisfying and falsifying $F$.


%
%
\begin{definition}
  \label{def:two_kinds_bp}
  Let $F(X,Y)$ be a formula and $G$ be a non-empty subset of $F$. Let
  (\pnt{x},\pnt{y}) be a $G$-boundary point of $F$ where \pnt{x} and
  \pnt{y} are full assignments to $X$ and $Y$ respectively.  The
  $G$-boundary point (\pnt{x},\pnt{y}) is called \bm{Y}\tb{-removable}
  (respectively \bm{Y}\tb{-unremovable}) if formula \cof{F}{y} is
  unsatisfiable (respectively satisfiable).
\end{definition}

Recall that \cof{F}{y} describes the formula $F$ in subspace
\pnt{y}. So the fact that \cof{F}{y} is unsatisfiable (or satisfiable)
just means that $F$ is unsatisfiable (respectively satisfiable) in
subspace \pnt{y}.  We use the name ``$Y$-removable boundary point''
since such a boundary point can be eliminated by adding a clause
implied by $F$ that depends only on variables of $Y$. Indeed, suppose
that (\pnt{x},\pnt{y}) is a $Y$-removable $G$-boundary point. Then
\cof{F}{y} is unsatisfiable and hence there is a clause $C(Y)$
falsified by \pnt{y} and implied by $F$. Note that (\pnt{x},\pnt{y})
is \ti{not} a $G$-boundary point of $F \cup \s{C}$ because it
\ti{falsifies} the formula $(F \cup \s{C}) \setminus G$. So, adding
$C$ to $F$ eliminates the $G$-boundary point (\pnt{x},\pnt{y}).  On
the contrary, a $Y$-\ti{unremovable} boundary point (\pnt{x},\pnt{y})
\ti{cannot} be eliminated by adding a clause falsified by \pnt{y} and
implied by $F$.

%
%
\begin{proposition}
 \label{prop:form_red}
  Let $F(X,Y)$ be a formula. Let $G$ be a non-empty subset of clauses of $G$.
  The formula $G$ is redundant in \prob{X}{F} if and only if every
  $G$-boundary point of $F$ (if any) is $Y$-unremovable.
\end{proposition}

So, to check the second condition of Proposition~\ref{prop:main} one
needs to show that every $G$-boundary point of $H \wedge F$ (if any)
is $Y$-unremovable.

\section{Description of \Vp}
\label{sec:ver_alg}
%
%
\setlength{\intextsep}{4pt}
\setlength{\textfloatsep}{4pt}
\begin{wrapfigure}{l}{2in}
\centering
\small
\vspace{-10pt}
\parbox{0cm}{\begin{tabbing}
aaa\=b\=cc\= dd\= \kill
$\Vp(\prob{X}{F},G,H)$ \{ \\
\scriptsize{1}\>  for every $C \in H$  \{ \\
\scriptsize{2}\Tt  $\pnt{p} := \mi{Sat}(F\!\wedge\!\overline{C})$\\
\scriptsize{3}\Tt  if ($\pnt{p} \neq \mi{nil}$)\\
\scriptsize{4}\ttt   return(\ti{false})\} \\
~~~~~------------\\
\scriptsize{5}\> for every $C \in G$ \{ \\
\scriptsize{6}\Tt $\mi{ok}\!:=\!\mi{ChkRed}(\prob{X}{F\!\wedge\!H},C)$ \\
\scriptsize{7}\Tt if ($\mi{ok}=\mi{false}$) return(\ti{false}) \\
\scriptsize{8}\Tt $F := F \setminus \s{C}$ \} \\
\scriptsize{9}\> return(\ti{true})  \\
\end{tabbing}}
\vspace{-20pt}
\caption{\Vp}
\vspace{-10pt}
\label{fig:ver_pqe}
\end{wrapfigure}

In this section, we describe the algorithm for verification of PQE
called \Vp.
%
%
\subsection{High-level view of \Vp}
A high-level view of \Vp is given in Fig.~\ref{fig:ver_pqe}. \Vp
accepts formula \prob{X}{F}, a subset $G \in F$ of clauses to take out
of the scope of quantifiers and a solution $H$ to this PQE
problem. That is \prob{X}{F} is supposed to be logically equivalent to
$H \wedge \prob{X}{F \setminus G}$. \Vp returns \ti{true} if this
equivalence holds and so, $H$ is a correct solution. Otherwise, \Vp
returns \ti{false}.

\Vp consists of two parts separated by a solid line. In the first part
(lines 1-4), \Vp just checks if $H$ is implied by $F$. This is done by
checking for every clause $C \in H$ if $F \wedge \overline{C}$ is
satisfiable. If so, $C$ is not implied by $F$ and the solution $H$ is
incorrect. Hence, \Vp returns \ti{false} (line 4). In the second part
(lines 5-9), for every clause $C \in G$, the algorithm checks if $C$
is redundant in \prob{X}{F \wedge H} by calling the function \Cr (line
6). Namely, \Cr checks if $\prob{X}{F \wedge H} \equiv \prob{X}{(F
  \setminus \s{C}) \wedge H}$. If so, $C$ is removed from $F$ (line
8). Otherwise, $C$ is not redundant in \prob{X}{F \wedge H} and \Vp
returns \ti{false} (line 7). If all clauses of $G$ can be removed from
\prob{X}{F \wedge H}, then $H$ is a correct solution and \Vp returns
\ti{true}.
%
%
\subsection{Description of \Cr}
%
%
\setlength{\intextsep}{4pt}
\setlength{\textfloatsep}{4pt}
\begin{wrapfigure}{l}{2in}
\centering
\small
\vspace{-5pt}
\parbox{0cm}{\begin{tabbing}
aaa\=b\=cc\= dd\= \kill
$\Cr(\prob{X}{F\wedge H},C)$ \{ \\
\scriptsize{1}\>  $Y := \V{F} \setminus X$ \\
\scriptsize{2}\> $\pl := \emptyset$  \\
\scriptsize{3}\> while (\ti{true}) \{ \\
\scriptsize{4}\Tt $F':=(F \setminus \s{C}) \wedge H$ \\
\scriptsize{5}\Tt $(\pnt{x},\pnt{y})\!:=\!\mi{Sat}(\pl\!\wedge F' \wedge\!\overline{C})$  \\
\scriptsize{6}\Tt if ($(\pnt{x},\pnt{y}) = \mi{nil}$) \\
\scriptsize{7}\ttt return(\ti{true}) \\
\scriptsize{8}\Tt $\pnt{x}^* := \mi{Sat}(\cof{F}{y} \wedge \cof{H}{y})$\\
\scriptsize{9}\Tt if ($\pnt{x}^* = \mi{nil}$)  \\
\scriptsize{10}\ttt  return(\ti{false}) \\
\scriptsize{11}\Tt  $D \!:=\!PlugCls(\pnt{y},\!\pnt{x}^*,\!F,\!H)$ \\
\scriptsize{12}\Tt $\pl := \pl \cup \s{D}$\}~~~\}\} \\
\end{tabbing}}
\vspace{-15pt}
\caption{\Cr}
\label{fig:chk_red}
\end{wrapfigure}
 The pseudocode of \Cr is shown in
Fig.~\ref{fig:chk_red}. \Cr accepts the formula \prob{X}{F(X,Y) \wedge
  H(Y)} and a quantified clause $C$ to be checked for redundancy. \Cr
returns \ti{true} if $C$ is redundant in \prob{X}{F \wedge
  H}. Otherwise, it returns false. To verify redundancy of $C$, \Cr
checks if $F \wedge H$ has a $Y$-removable $C$-boundary point. If not,
$C$ is redundant. Otherwise, $C$ is not redundant.

\Cr starts with computing the set $Y$ of unquantified variables (line
1). Then it initializes the set of ``plugging'' clauses (see below).
The main work is done in the while loop (lines 3-12). \Cr starts with
checking if formula $F \wedge H$ has a $C$-boundary point (lines 4-5)
i.e.  checking if there is an assignment (\pnt{x},\pnt{y}) satisfying
$\mi{Plg} \wedge (F \setminus \s{C}) \wedge H \wedge
\overline{C}$. The formula \ti{Plg} is used here to exclude the
$C$-boundary points examined in the previous iterations of the
loop. If no (\pnt{x},\pnt{y}) exists, the clause $C$ is redundant and
\Cr returns \ti{true} (line 7).

If the assignment (\pnt{x},\pnt{y}) above exists, \Cr checks if
formula $\cof{F}{y} \wedge \cof{H}{y}$ is satisfiable i.e. whether $F
\wedge H$ is satisfiable in subspace \pnt{y} (lines 8-10). If not, the
$C$-boundary point (\pnt{x},\pnt{y}) is $Y$-removable. This means that
$C$ is not redundant in \prob{X}{F \wedge H} and \Cr returns
\ti{false} (line 10). Otherwise, (\pnt{x},\pnt{y}) is a
$Y$-unremovable boundary point and \Cr calls the function \ti{PlugCls}
to build a plugging clause $D(Y)$. The latter is falsified by \pnt{y}
and so excludes re-examining $C$-boundary points in the subspace
\pnt{y}. After that, \Cr adds $D$ to the formula \ti{Plg} and starts a
new iteration of the loop.

The simplest way to build $D$ is to form the longest clause falsified
by \pnt{y}. One can try to make $D$ shorter to exclude a greater
subspace from future considerations. Suppose there is $\pnt{y}^*
\subset \pnt{y}$ such that the assignment $\pnt{x}^*$ satisfying
$\cof{F}{y} \wedge \cof{H}{y}$ (found in line 8) still satisfies
$F_{\vec{y}^*} \wedge H_{\vec{y}^*} $. This means that every
$C$-boundary point of the larger subspace $\pnt{y}^*$ is
$Y$-unremovable too. So, one can add to \ti{Plg} a shorter plugging
clause $D$ falsified by $\pnt{y}^*$ rather than \pnt{y}.

%
%

\subsection{Scalability issues}
\label{ssec:perf}
As we mentioned earlier, \Vp consists of two parts. The first part of
\Vp checks if every clause of the solution $H$ is implied by $F$. In
the second part, for every clause $C \in G$, the function \Cr checks
if $C$ is redundant in \prob{X}{F \wedge H}. (Recall that $G$ is the
subset of clauses that one must take out of \prob{X}{F}.) The first
part reduces to $|H|$ calls to a SAT-solver. So, it is as scalable as
SAT-solving (unless $H$ blows up as the size of the PQE problem
grows).  The second part of \Vp scales much poorer. The reason is that
this part requires enumeration of $Y$-unremovable $C$-boundary points
and the number of such points is typically grows exponentially. Besides,
the plugging clauses produced by \Cr are long.  So, adding a plugging
clause cannot exclude a big chunk of $C$-boundary points at once. So,
the size of formulas that can be efficiently handled by \Vp is limited
by 70-80 variables.

There is however \tb{an important case} where \Vp can efficiently
verify large formulas. This is the case where $F \wedge H$ does not
have any $Y$-unremovable $G$-boundary points. (Since $F$ and $F \wedge
H$ have the same $Y$-unremovable $G$-boundary points, this means that
$F$ has no such boundary points either.)  Then for every clause $C \in
G$, the \smallskip function \Cr immediately finds out that the formula
$(F \setminus \s{C}) \wedge H \wedge \overline{C}$ is unsatisfiable.
So, the verification of solution $H$ reduces to $|H| + |G|$
SAT-checks.

\section{Experimental Results}
\label{sec:expers}
In this section, we experimentally evaluate our implementation of \Vp.
In this implementation, we used Minisat~\cite{minisat} as an internal
SAT-solver. The source of \Vp and some examples can be downloaded
from~\cite{ver_pqe}. We conducted three experiments in which we
verified solutions obtained by the PQE algorithm called
\egp~\cite{eg_pqe_tech}. In the experiments we solved the PQE problem
of taking a clause $C$ out of formula \prob{X}{F(X,Y)} i.e.  finding a
formula $H(Y)$ such that $\prob{X}{F} \equiv H \wedge \prob{X}{F
  \setminus \s{C}}$. In Subsections~\ref{ssec:rmv_bps}
and~\ref{ssec:unrem_bps} we consider large formulas appearing in the
process of ``property generation''.  Namely, these formulas were
constructed when generating properties of circuits from the HWMCC-13
set as described in~\cite{eg_pqe_tech}. In
Subsection~\ref{ssec:rand_form}, we consider small random formulas. In
the experiments we used a computer with
Intel\textsuperscript{\textregistered} Core\textsuperscript{TM}
i5-10500 CPU @ 3.10GHz.

\subsection{Formulas where all boundary points are removable}
\label{ssec:rmv_bps}

In this subsection, we consider the PQE problems of taking $C$ out of
\prob{X}{F(X,Y)} where all $C$-boundary points of $F$ are
$Y$-removable.  In~\cite{eg_pqe_tech}, we generated 3,736 of such
formulas. In Table~\ref{tbl:rmv_bps}, we give a sample of 7 formulas.
The first column of the table gives the name of the circuit of the
HWMCC-13 set used to generate the PQE problem. (The real names of
circuits \ti{exmp1}, \ti{exmp2} and \ti{examp3} in the HWMCC-13 set
are \ti{mentorbm1}, \ti{bob12m08m}, and \ti{bob12m03m} respectively.)

%
%
\begin{wraptable}{l}{2.8in}
\centering
\scriptsize
\captionsetup{justification=centering}
\caption{\small{\Vp on formulas where all boundary points are removable}}
  \begin{tabular}{|p{25pt}|p{28pt}|p{25pt}|p{21pt}|p{17pt}|p{30pt}|p{29pt}|} \hline
 name & cla- & vari- &size & size &\tiny\ti{EG-PQE}$^+$& \tiny\ti{VerPQE}\\
 of  & uses  & ables & of set & of $H$ & run & run\\ 
 circ. & of $F$  & of $F$       &~~$Y$   &    &time\,(s)  & time\,(s)  \\ \hline
exmp1     & 64,365  &  26,998  & 4,376 &~1 &~~ 0.2  &~~ 0.03 \\ \hline
6s207   & 73,457  &  30,540  & 3,012 &~6  & ~~ 0.2  &~~ 0.04     \\ \hline
exmp2     & 84,009  &  32,147  & 1,994 &~1 & ~~ 0.1  &~~ 0.1  \\ \hline
exmp3     & 94,523  &  41,354  & 5,174 &~825 & ~~ 11   &~~ 0.4  \\ \hline
6s249   & 226,666 &  78,289  & 1,111 &~1 & ~~ 0.4  &~~ 0.1  \\ \hline
6s428   & 231,506 &  92,274  & 3,790 &~118 &~~  2.8  &~~ 0.2  \\ \hline
6s311   & 259,086 &  87,974  & 519   &~80 & ~~ 2.0  &~~ 0.1  \\ \hline
\end{tabular}                
\label{tbl:rmv_bps}
\end{wraptable}

The second and third co-lumns give the number of clauses and variables
of formula $F$. The fourth column shows the size of the set $Y$ i.e.
the number of unquantified variables in \prob{X}{F(X,Y)}. The next
column gives the number of clauses in the solution $H$ found by \egp.
The last two columns show the time taken by \egp and \Vp (in seconds)
to finish the PQE problem and verify the solution. As we mentioned in
Subsection~\ref{ssec:perf}, if all $C$-boundary points of $F$ are
$Y$-removable the same applies to formula $F \wedge H$. So, \Vp should
be very efficient even for large formulas.  Table~\ref{tbl:rmv_bps}
substantiates this intuition.

\subsection{Formulas with unremovable boundary points}
%
%
\vspace{5pt}
\begin{wraptable}{l}{2.8in}
\centering
\scriptsize
\captionsetup{justification=centering}
\caption{\small{\Vp on formulas with unremovable boundary points. The
time limit is 600 sec.}}
  \begin{tabular}{|p{25pt}|p{28pt}|p{25pt}|p{20pt}|p{17pt}|p{30pt}|p{29pt}|} \hline
 name    & cla-     & vari-  &size     & size &\tiny\ti{EG-PQE}$^+$ & \tiny\ti{VerPQE}\\
 of    & uses     & ables  & of set  & of $H$ & run  & run\\ 
 circ.   &  of $F$    & of $F$       &~$Y$    & & time\,(s) & time\,(s)  \\ \hline
 6s209   &  25,086  &  14,868  & 5,759 &~5 &~~ 0.1  & $>$600 \\ \hline 
 6s413   &  29,321  & 14,063   & 4,343 &~18 &~~ 0.2  & $>$600 \\ \hline
 6s276   &  35,810  & 17,631   & 3,201 &~11 &~~ 0.1  & $>$600 \\ \hline
 6s176   &  39,704  & 15,754   & 1,566 &~0 &~~ 0.9  & $>$600 \\ \hline
 6s207   &  73,457  & 30,540   & 3,012 &~20 &~~ 0.5  & $>$600 \\ \hline
 6s110   &  83,396  & 34,165   & 807   &~6 &~~ 0.2  & ~~0.1    \\ \hline
 6s275   &  109,328 & 49,130   & 3,196 &~2 &~~ 0.1  & $>$600 \\ \hline

\end{tabular}                
\label{tbl:unrem_bps}
\end{wraptable}

Here we consider the same PQE problems as in the previous
subsection. The only difference is that the formula $F$ contains
$C$-boundary points that are $Y$\ti{-unremovable}.
In~\cite{eg_pqe_tech}, we generated 3,094 of such formulas. In
Table~\ref{tbl:unrem_bps}, we give a sample of 7 formulas.  The name
and meaning of each column is the same as in Table~\ref{tbl:rmv_bps}.

\label{ssec:unrem_bps}

Table~\ref{tbl:unrem_bps} shows that \Vp failed to verify 6 out of 7
solutions in the time limit of 600 sec. whereas the corresponding
problems were easily solved by \egp. (The reason is that \egp uses a
more powerful technique of proving redundancy of $C$ than plugging
unremovable boundary points as \Vp does.) So, solutions $H$ obtained
for large formulas \prob{X}{F} where $F$ has a lot of unremovable
boundary points cannot be efficiently verified by \Vp.

\subsection{Random formulas}
\label{ssec:rand_form}
 In this subsection, we continue consider formulas with
 $Y$-unremovable $C$-boundary points. Only, in contrast to the
 previous subsection, here we consider small random formulas. In this
 experiment we verified solutions obtained for formulas whose number
 of variables ranged from 70 to 85. To get more reliable data, for
 each size we generated 100 random PQE problems and computed the
 average result. For each example, the formula $F$ had 20\% of
 two-literal and 80\% of three-literal clauses.

%
%
\begin{wraptable}{l}{2.5in}
\centering
\scriptsize
\captionsetup{justification=centering}
\caption{\small{\Vp on random formulas}}
\vspace{-5pt}
  \begin{tabular}{|p{25pt}|p{15pt}|p{17pt}|p{20pt}|p{18pt}|p{30pt}|p{29pt}|} \hline
 num- & cla- & vari- &size & size&\tiny\ti{EG-PQE}$^+$& \tiny\ti{VerPQE}\\
 ber of  &uses   & ables  & of set & of $H$  & run & run\\ 
 prob. & of $F$ &  of $F$      & ~$Y$  &     &time\,(s)  & time\,(s)  \\ \hline
 ~100 & 140     &~70  &~35 &~28 &~~ 0.01  &~~ 1.0  \\ \hline
  ~100 & 150     &~75   &~37 &~41 &~~ 0.01 &~~ 4.7  \\ \hline
   ~100 & 160     &~80  &~40  &~69 &~~ 0.03 &~~ 11.5  \\ \hline
    ~100 & 170     &~85   &~42 &~63 &~~ 0.03  &~~ 98.3   \\ \hline

\end{tabular}                
\vspace{5pt}
\label{tbl:rf}
\end{wraptable}

The results of this experiment are shown in Table~\ref{tbl:rf}. Let us
explain the meaning of each column of this table using its first line.
The first column indicates that we generated 100 PQE problems of the
same size shown in the next three columns. That is for all 100
problems corresponding to the first line of Table~\ref{tbl:rf} the
number of clauses, variables and the size of the set $Y$ was 140, 70
and 35 respectively. The last three columns of the first line show the
\ti{average} results over 100 examples. For instance, the first column
of the three says that the average size of the solution $H$ found by
\egp was 28 clauses. Table~\ref{tbl:rf} shows that the performance of
\Vp drastically drops as the number of variables grows due to the
exponential blow-up of the set of $Y$-unremovable $C$-boundary points.

\section{Some Background}
\label{sec:bkgr}
In this section, we give some background on boundary points.  The
notion of a boundary point with respect to a variable was introduced
in~\cite{esspnts}. (At the time it was called an \ti{essential}
point).  Given a formula $F(X)$, a boundary point with respect to a
variable $x \in X$ is a full assignment \pnt{p} to $X$ such that each
clause falsified by \pnt{p} contains $x$. Later we showed a relation
between a resolution proof and boundary points~\cite{sat09}. Namely,
it was shown that if $F$ is unsatisfiable and contains a boundary
point with respect to a variable $x$, any resolution proof that $F$ is
unsatisfiable has to contain a resolution on $x$. In~\cite{hvc-10}, we
presented an algorithm that performs SAT-solving via boundary point
elimination.

In~\cite{tech_rep_edpll,fmsd14}, we introduced the notion of a
boundary point with respect to a subset of variables rather than a
single variable. Using this notion we formulated a QE algorithm that
builds a solution by eliminating removable boundary
points. In~\cite{eg_pqe_tech}, we formulated two PQE algorithms called
\Eg and \egp.  The algorithm \Eg is quite similar to \Vp and
implicitly employs the notion of a boundary point we introduced here
i.e. the notion formulated with respect to a subset of \ti{clauses}
rather than variables. In this report, when describing \Vp we use this
notion of a boundary point \ti{explicitly}.

\section{Conclusions}
\label{sec:concl}
We present an algorithm called \Vp for verifying a solution to Partial
Quantifier Elimination (PQE). The advantage of \Vp is that it does
need to know how this solution was obtained (e.g. if a particular
proof system was employed).  So, \Vp can be used to debug an \ti{any}
PQE algorithm.  A flaw of \Vp is that its performance strongly depends
on the presence of so-called unremovable boundary points of the
formula at hand. If this formula has no such points, \Vp can
efficiently verify solutions to very large PQE problems. Otherwise,
its performance is, in general, limited to small problems of 70-80
variables.

\bibliographystyle{IEEEtran}
\bibliography{short_sat,local,l1ocal_hvc}
\vspace{20pt}
\appendix
\noindent{\large \tb{Appendix}}
\section{Proofs Of Propositions}
\label{app:proofs}
\setcounter{proposition}{0}
%
%
\begin{proposition}
Let $H$ be a solution to the PQE problem of
Definition~\ref{def:pqe_prob}.  That is $\prob{X}{F}\equiv
H\wedge\prob{X}{F \setminus G}$. Then $F \imp H$ (i.e. $F$ implies
$H$).
\end{proposition}
%
%
\begin{proof}
By conjoining both sides of the equality with $H$ one concludes
that\linebreak $H \wedge \prob{X}{F} \equiv
H\wedge\prob{X}{F \setminus G}$, which entails
$H \wedge \prob{X}{F} \equiv \prob{X}{F}$. Then\linebreak $\prob{X}{F} \imp H$
and thus $F \imp H$.
\end{proof}
%
%
\begin{proposition}
Formula $H(Y)$ is a solution to the PQE problem of taking $G$ out
of the scope of quantifiers in \prob{X}{F(X,Y)} if and only if
  \begin{enumerate}
  \item[a)]  $H$ is implied by $F$;
  \item[b)]  $G$ is redundant in $H \wedge \prob{X}{F}$ i.e.
    $H \wedge \prob{X}{F} \equiv H \wedge \prob{X}{F \setminus G}$
 \end{enumerate}
\end{proposition}
%
%
\begin{proof}\noindent\tb{The if part.} Given the two conditions above,
one needs to prove that $\prob{X}{F} \equiv
H \wedge \prob{X}{F \setminus G}$.  Assume the contrary i.e.
$\prob{X}{F} \not\equiv H \wedge \prob{X}{F \setminus G}$.  Consider
the two possible cases. The first case is that there exists a full
assignment \pnt{y} to $Y$ such that $F$ is satisfiable in
subspace \pnt{y} whereas $H \wedge (F \setminus G)$ is unsatisfiable
in this subspace. Since $F \setminus G$ is satisfiable in the
subspace \pnt{y}, $H$ is unsatisfiable in this subspace.  So, $F$ does
not imply $H$ and we have a contradiction.

The second case is that $F$ is unsatisfiable in the subspace \pnt{y}
whereas\linebreak $H \wedge (F \setminus G)$ is satisfiable
there. Then $H \wedge F$ is unsatisfiable in subspace \pnt{y} too.
So, $H \wedge \prob{X}{F} \neq H \wedge \prob{X}{F \setminus G}$ in
subspace \pnt{y} and hence $G$ is not redundant in
$H \wedge \prob{X}{F}$. So, we have a contradiction again.

\vspace{5pt}
\noindent\tb{The only if part}. Given
$\prob{X}{F} \equiv H \wedge \prob{X}{F \setminus G}$, one needs to
prove the two conditions above. The first condition (that $H$ is
implied by $F$) follows from Proposition~\ref{prop:sol_impl}.  Now
assume that the second condition (that $G$ is redundant in
$H \wedge \prob{X}{F}$) does not hold. That is
$H \wedge \prob{X}{F} \not\equiv H \wedge \prob{X}{F \setminus G}$.
Note that if $H \wedge F$ is satisfiable in a subspace \pnt{y}, then
$H \wedge (F \setminus G)$ is satisfiable too. So, the only case to
consider here is that $H \wedge F$ is unsatisfiable in a
subspace \pnt{y} whereas $H \wedge (F \setminus G)$ is satisfiable
there.  This means that $F$ is unsatisfiable in the
subspace \pnt{y}. Then $\prob{X}{F} \neq H \wedge \prob{X}{F \setminus
G}$ in this subspace and we have a contradiction.
\end{proof}
%
%
\begin{proposition}
  Let $F(X,Y)$ be a formula. Let $G$ be a non-empty subset of clauses of $G$.
  The formula $G$ is redundant in \prob{X}{F} if and only if every
  $G$-boundary point of $F$ (if any) is $Y$-unremovable.
\end{proposition}
%
%
\begin{proof}
\noindent\tb{The if part.} Given that every $G$-boundary point of $F$ 
is $Y$-unremovable, one needs to show that $G$ is redundant
in \prob{X}{F} i.e.  $\prob{X}{F} \equiv \prob{X}{F \setminus
G}$. Assume that this is not true.  Then there is a full
assignment \pnt{y} to $Y$ such that $F$ is unsatisfiable in
subspace \pnt{y} whereas $F \setminus G$ is satisfiable there. This
means that there is an assignment (\pnt{x},\pnt{y}) falsifying $F$ and
satisfying $F \setminus G$. Since this assignment falsifies $G$, it is
a $G$-boundary point. This boundary point is $Y$-\ti{removable},
because $F$ is unsatisfiable in subspace \pnt{y}. So, we have a
contradiction.

\vspace{5pt}
\noindent\tb{The only if part}. Given that $G$ is redundant
in \prob{X}{F}, one needs to show that every $G$-boundary point of $F$
is $Y$-unremovable. Assume the contrary i.e. there is a $Y$-removable
$G$-boundary point of $F$. This means that there is an assignment
(\pnt{x},\pnt{y}) falsifying $G$ and satisfying $F \setminus G$ such
that $F$ is unsatisfiable in subspace \pnt{y}.  Then
$\prob{X}{F} \neq \prob{X}{F \setminus G}$ in subspace \pnt{y} and so,
we have a contradiction.
\end{proof}

\end{document}